\documentclass{article}
\usepackage{graphicx}
\usepackage{amsfonts}
\usepackage{amsmath}
\usepackage{amsthm}
\usepackage{subcaption}
\usepackage{color}
\usepackage{url}
\usepackage{hyperref}

\newcounter{mycount}[section]

\newtheorem{theorem}[mycount]{Theorem}
\newtheorem{lemma}[mycount]{Lemma}
\newtheorem{definition}[mycount]{Definition}

\newtheorem{corollary}[mycount]{Corollary}

\title{Hyperbolic Minesweeper is in P}

\author{Eryk Kopczy\'nski}

\def\bbH{\mathbb{H}}
\def\area{\mathrm{area}}
\def\perimeter{\mathrm{peri}}

\def\subfig#1#2#3{
\begin{minipage}[b]{#1} \noindent
  \includegraphics[width=\textwidth]{#2}
  \subcaption{#3}
\end{minipage}
}

\begin{document}

\maketitle

\begin{abstract}
We show that, while Minesweeper is NP-complete, its hyperbolic variant is in P. 
Our proof does not rely on the rules of Minesweeper, but is valid for any puzzle
based on satisfying local constraints on a graph embedded in the hyperbolic plane.
\end{abstract}

\section{Introduction}

\begin{figure}[ht]
\begin{center}
\subfig{0.48\linewidth}{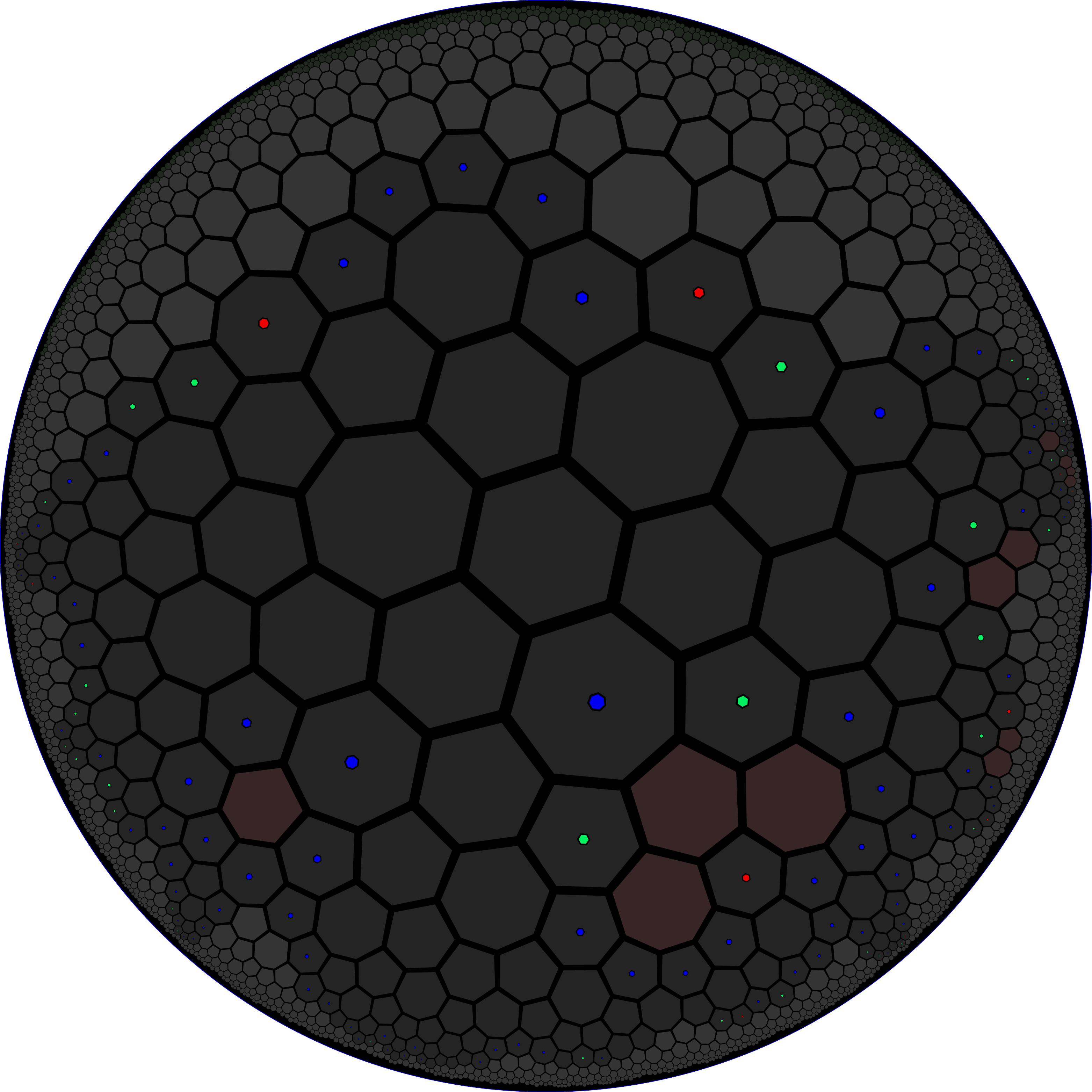}{graphical mode, bitruncated tessellation}
\subfig{0.48\linewidth}{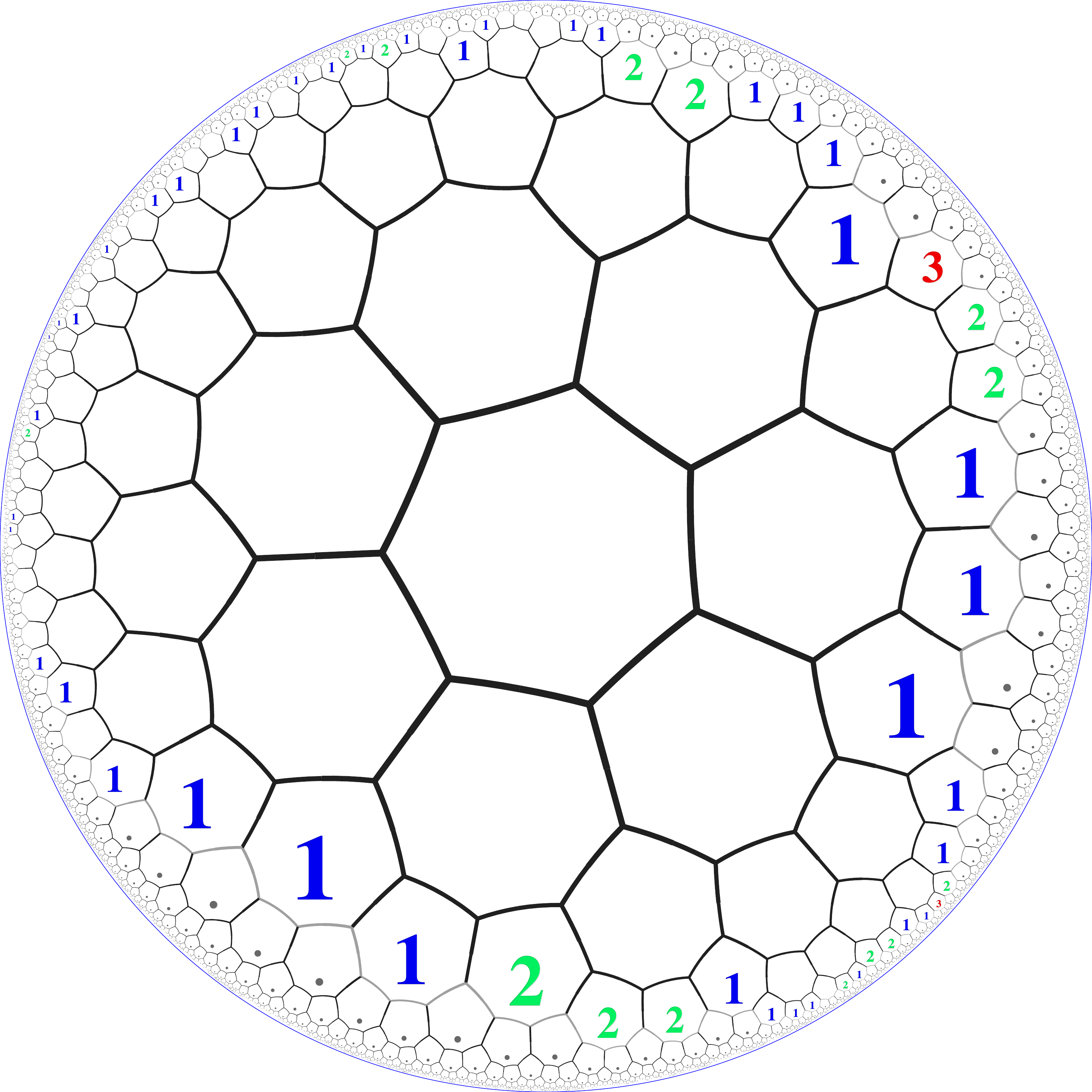}{text mode, regular tessellation}
\end{center}
\caption{\label{hrhex} Hyperbolic Minesweeper \cite{hyperrogue}. In (a), the default settings are used
(bitruncated order-3 heptagonal tessellation, numbers of adjacent mines are color-coded;
some of the mines that the players is sure of are marked red).
In (b) we play on an order-3 heptagonal tessellation, and numbers are shown.}
\end{figure}

Minesweeper is a popular game included with many computer systems; it also exists in the puzzle form. In the puzzle form,
every cell in a square grid either contains a number or is empty. The goal of the puzzle is to assign mines to the
empty squares in such a way that every number $n$ is adjacent (orthogonally or diagonally) to exactly $n$ mines.

This puzzle is a well-known example of a NP-complete problem \cite{minenp}. Its popularity has also spawned many
variants played on different grids, from changing the tessellation of the plane, to changing the number of dimensions
(six-dimensional implementations of Minesweeper exist) to changing the underlying geometry.

In this paper we will be playing Minesweeper on a tessellation of the hyperbolic plane. Minefield, a land based on hyperbolic
Minesweeper is available in HyperRogue \cite{hyperrogue}; this implementation differs from the standard Minesweeper in multiple ways,
e.g., by being played on an infinite board, however, the basic idea is the same. Another implementation is Warped Mines
for iOS \cite{hypersweeper}, which has the same rules as the usual Minesweeper except the board, which is a bounded subset of
the order-3 heptagonal tiling of the hyperbolic plane.

From the point of view of a computer scientist,
the most important distinctive property of hyperbolic geometry is exponential growth: the area of a circle
of radius $r$ grows exponentially with $r$. It is also more difficult to understand than Euclidean geometry.
While these properties often cause computational geometry problems to be more difficult, it also gives hyperbolic geometry
applications in data visualization \cite{lampingrao,munzner} and data analysis \cite{papa}. 

In this paper we show that hyperbolic geometry makes Minesweeper easier: 
the hyperbolic variant
of Minesweeper is in P. Our proof will not rely on the specific rules of Minesweeper nor work with any specific 
tessellation of the hyperbolic plane; instead, it will work with any puzzle based on satisfying local constraints,
on any graph that naturally embeds into the hyperbolic plane.

\section{Hyperbolic Geometry}
We denote the hyperbolic plane with $\bbH^2$. 
Since this paper requires only the basic understanding of hyperbolic geometry, we will not include the formal definition;
see \cite{cannon}, or \cite{hyperrogue} for an intuitive introduction.
Figures \ref{hrhex} and \ref{egrowth} shows the hyperbolic plane in the Poincar\'e disk model, which is a projection which
represents angles faithfully, but not the distances: the scale gets smaller and smaller as we get closer to the circle
bounding the model. In particular, all the heptagons in Figure \ref{hrhex}b are the same size, all the heptagons in 
Figure \ref{hrhex}a are the same size, and all the hexagons in Figure \ref{hrhex}a are the same size.
We denote the distance between two points $x,y \in \bbH^2$ by $\delta(x, y)$.
The set of points in distance at most $r$ from $x$ is denoted with $B(x,r)$. The area of $B(x,r)$ is $2 \pi (\cosh r - 1)$, which
we denote with $\area(r)$. The perimeter of $B(x,r)$ is $2\pi\sinh r$, which we denote with $\perimeter(r)$. Note that both
$\area$ and $\perimeter$ grow exponentially: $\area(r) = \Theta(e^r)$ and $\perimeter(r) = \Theta(e^r)$.

\begin{figure}[ht]
\begin{center}
\includegraphics[width=0.48\linewidth]{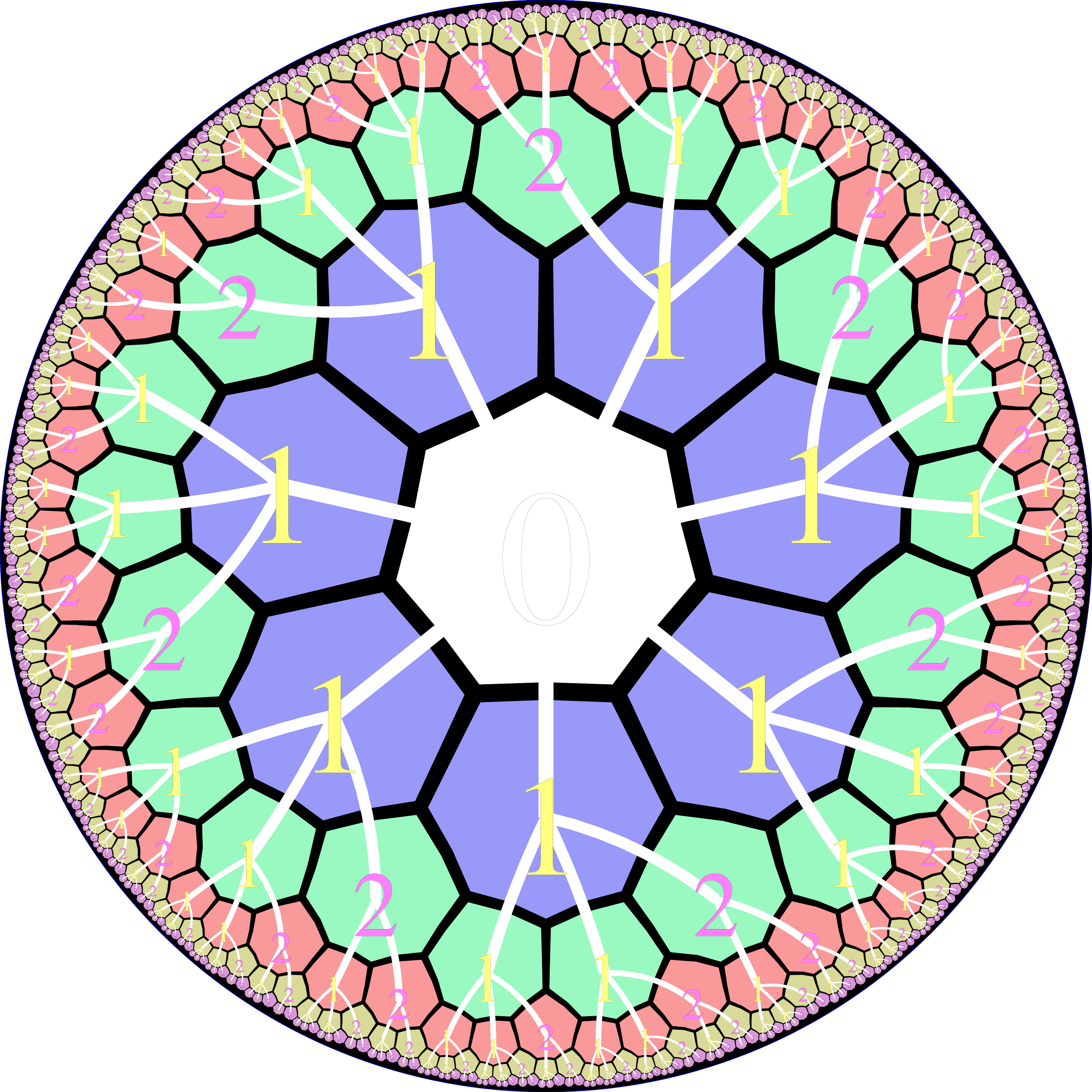}
\end{center}
\caption{\label{egrowth} Exponential growth in the hyperbolic plane.}
\end{figure}

We include Figure \ref{egrowth} as a simple introduction to the structure of the hyperbolic plane and its exponential growth.
This is the order-3 heptagonal tessellation \cite{margenstern_heptagrid}. The numbers correspond
to the number of adjacent tiles which are closer to the center tile. This tessellation can be constructed as follows.
We construct a tree of tiles:
the central tile has seven children of type 1;
each tile of tile 1 has two children of type 1, and one child
of type 2; each tile of tile 2 has one child of type 1 and one child of type 2. To
produce the actual tessellation, we also connect adjacent tiles on the same level, and to the leftmost child of the tile on the right.
The hyperbolic plane can be seen as a continuous version of the graph obtained using this construction. Since every tile has at least
two children, there are exponentially many tiles in $r$ steps from the central tile. 

\section{Hyperbolic Graph}
We need a suitable definition of a hyperbolic graph. The obvious definition, obtained by changing ``plane'' to ``hyperbolic plane'' in
the definition of a planar graph, is equivalent to the definition of the planar graph (the plane and the hyperbolic plane are homeomorphic).
Another definition found in literature is the notion of Gromov $\delta$-hyperbolic graph \cite{gromovhyp}; this definition is based
on the observation that, in the hyperbolic plane, triangles are thin, i.e., for any three vertices $a$, $b$, $c$, every point on
any shortest path from $a$ to $c$  must be in distance at most $\delta$ from either the shortest path from $a$ to $b$, or the shortest
path from $b$ to $c$. The parameter $\delta$ measures tree-likeness of the graph (it can be easily seen that trees are 0-hyperbolic). 
However, this definition is not suitable
for us, because every graph $G$ becomes 2-hyperbolic when we add a vertex $v_*$ connected to every $v \in V(G)$; our main result is not true
for such graphs. We propose the following definition:

\begin{definition}
A $(r,d)$-{\bf hyperbolic graph} is a graph $G = (V,E)$ such that there exists an embedding $m: V \rightarrow \bbH^2$ such that:
\begin{itemize}
\item for $v_1 \neq v_2$, $\delta(m(v_1), m(v_2)) > r$,
\item for $\{v_1, v_2\} \in E$, $\delta(m(v_1), m(v_2)) < d$,
\item if we draw every edge $\{v_1, v_2\} \in E$ as a straight line segment between $m(v_1)$ and $m(v_2)$, these edges do not cross,
nor they get closer to vertices other than $v_1$ or $v_2$ than $r/2$.
\end{itemize}
\end{definition}

Intuitively, $r/2$ is the radius of every vertex; this parameter bounds the density of our graph embedding.
The parameter $d$ gives the maximum distance between two vertices connected with an edge.

This definition includes finite subsets of all regular and semiregular hyperbolic tessellations (such as the ones shown in Figure \ref{hrhex}).
We denote $N(v)$ to be the neighborhood of $v \in V$,
i.e., $\{v\} \cup \{w \in V: \{v,w\} \in E\}$. (Our proof also works with neighborhoods of larger radius.)

{\bf Remark.} All $(r,d)$-hyperbolic graphs have degree bounded by a constant (for fixed $r,d$). Indeed, let $v \in V$.
For every $w \in N(v)$, the circles $B(w,\frac{r}{2})$ are disjoint, and they fit in
$B(v,d+\frac{r}{2})$. Therefore, $|N| \leq \area(d+\frac{r}{2}) / \area(\frac{r}{2})$.

\section{Hyperbolic Local Constraint Satisfaction Problem}

Below we state our main result.

\begin{theorem}\label{hlcsp}
Fix the set of colors $K$ and the parameters $(r,d)$. The following problem (Hyperbolic Local Constraint Satisfaction Problem, HLCSP)
can be solved in polynomial time:

{\bf INPUT:}
\begin{itemize}
\item a $(r,d)$-hyperbolic graph $G = (V,E)$;
\item for every vertex $v \in V$, $m(v)$, a subset of $K^{N(v)}$.
\end{itemize}

{\bf OUTPUT:} Is there a coloring $c: V \rightarrow K$ such that for every $v \in V$, $c _{| N(v)} \in m(v)$?
\end{theorem}

HLCSP generalizes hyperbolic minesweeper. We have two colors (no mine and mine), and for every $v \in V$ containing a number $k$, the constraint $m(v)$ says
that $v$ contains no mine itself, and exactly $k$ vertices in $N(v)$ contain a mine.

To prove Theorem \ref{hlcsp}, it is enough to prove that the following problem (HECSP) is in P.

\begin{theorem}
HLCSP reduces to the Hyperbolic Edge Constraint Satisfaction Problem (HECSP) given as follows:

{\bf INPUT:}
\begin{itemize}
\item a $(r,d)$-hyperbolic graph $G = (V,E)$;
\item for every edge $e \in E$, $m(e)$, a subset of $K^e$.
\end{itemize}

{\bf OUTPUT:} Is there a coloring $c: V \rightarrow K$ such that for every $e \in E$, $c _{|e} \in m(e)$?

(Note: this reduction changes the set of admissible colors $K$.)
\end{theorem}

\begin{proof} Let $(G,m)$ be the instance of HLCSP. We will be coloring $V$ using colors $K' = \{1, \ldots, k\}$, where $k$ is the greatest number of
elements of $m(v)$ (since $(r,d)$-hyperbolic graphs have bounded degree and $K$ is fixed, $k$ is also bounded).
Enumerate every elements of $m(v)$ with one color from $K'$. For $e = \{v_1, v_2\}\in E$, $m(e)$ is the set of all colorings $c: \{v_1, v_2\} \rightarrow K$ which are
consistent, i.e., $k_1$ denotes $c_1 \in m(v_1)$ and $k_2$ denotes $c_2 \in m(v_2)$, and $c_1$ equals $c_2$ on all the vertices in $N(v_1) \cap N(v_2)$.
\end{proof}

\section{Proof of Theorem \ref{hlcsp}}

We will be using the following result \cite{robseym,grigplanar}:
\begin{theorem}\label{robsey}
Given a planar graph $G = (V,E)$ and a number $t$, it is possible to either find a $t \times t$ grid as a minor of $G$, or produce a tree decomposition
of $G$ of width $\leq 5t-6$, in time $O(n^2 \log(n))$, where $n = |V|$.
\end{theorem}

\begin{definition}A tree decomposition of width $w$ of a graph $G=(V,E)$ is $(V_T, E_T, X)$ where $(V_T,E_T)$ is a tree, and
$X:V_T \rightarrow P(V)$ is a mapping which assigns
a subset of $V$ of cardinality at most $w+1$ to every vertex in $V_T$, such that:
\begin{itemize}
\item For every $v \in V$, the set of vertices $b \in V_T$ such that $v \in X_b$ is connected,
\item For every $e \in E$, there exists a $b \in V_T$ such that $e \subseteq X_b$.
\end{itemize}
\end{definition}

We can assume that our tree is rooted in $r \in V_T$. For $b \in V_T$, let $X^+_b$ be the union of $X_{b'}$ for all $b'$ which are descendants of $b$.

\begin{lemma}\label{cases}
Without loss of generality we can assume that every $b \in V_T$ falls into one of the following cases:

\begin{itemize}
\item $b$ is a leaf and $|X_b|=1$,
\item $b$ has a single child $b'$ and $X_{b'} = X_{b} \cup \{v\}$,
\item $b$ has a single child $b'$ and $X_{b'} = X_{b} - \{v\}$,
\item $b$ has two children $b_1$ and $b_2$, $X_b = X_{b_1} = X_{b_2}$, and $X^+_{b_1} - X_b$ and $X^+_{b_2} - X_b$ are non-empty and disjoint.
\end{itemize}
\end{lemma}

\begin{theorem}\label{nogrid}
If a $(r,d)$-hyperbolic graph $G$ contains a $t \times t$ grid as a minor, then $t = O(\log(V(G))$.
\end{theorem}
\begin{proof}
Without loss of generality we can assume $t=2k+1$. Such a $t \times t$ grid contains $k$ cycles around the center of the grid.
Let $v_c \in V$ be the vertex corresponding to this center. We have $k$ cycles $C_1, \ldots, C_k$ in the graph $V$, 
each of which surrounds $v_c$ and the preceding cycles. We use the following lemma:
\begin{lemma}\label{dlemma}
Every point in the drawing of cycle $C_i$ is in distance $\Omega(i)$ from $v$.
(The drawing is the polygon obtained as a union of the edges embedded in $\bbH^2$.)
\end{lemma}
Therefore, $|C_i| \geq \perimeter(\Omega(i))/d$. Since $\perimeter$ grows
exponentially, we get $k = O\log(|V|)$.
\end{proof}

\begin{figure}[t]
\begin{center}
\begingroup%
  \makeatletter%
  \providecommand\color[2][]{%
    \errmessage{(Inkscape) Color is used for the text in Inkscape, but the package 'color.sty' is not loaded}%
    \renewcommand\color[2][]{}%
  }%
  \providecommand\transparent[1]{%
    \errmessage{(Inkscape) Transparency is used (non-zero) for the text in Inkscape, but the package 'transparent.sty' is not loaded}%
    \renewcommand\transparent[1]{}%
  }%
  \providecommand\rotatebox[2]{#2}%
  \newcommand*\fsize{\dimexpr\f@size pt\relax}%
  \newcommand*\lineheight[1]{\fontsize{\fsize}{#1\fsize}\selectfont}%
  \ifx\svgwidth\undefined%
    \setlength{\unitlength}{121.5bp}%
    \ifx\svgscale\undefined%
      \relax%
    \else%
      \setlength{\unitlength}{\unitlength * \real{\svgscale}}%
    \fi%
  \else%
    \setlength{\unitlength}{\svgwidth}%
  \fi%
  \global\let\svgwidth\undefined%
  \global\let\svgscale\undefined%
  \makeatother%
  \begin{picture}(1,0.90075957)%
    \lineheight{1}%
    \setlength\tabcolsep{0pt}%
    \put(0,0){\includegraphics[width=\unitlength,page=1]{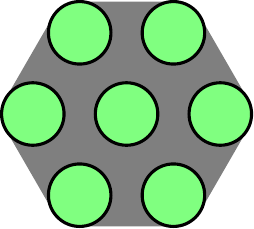}}%
    \put(0.5,0.41334275){\color[rgb]{0,0,0}\makebox(0,0)[t]{\lineheight{1.25}\smash{\begin{tabular}[t]{c}$v_c$\end{tabular}}}}%
    \put(0,0){\includegraphics[width=\unitlength,page=2]{grow.pdf}}%
    \put(0.52897135,0.82132523){\color[rgb]{0,0,0}\makebox(0,0)[t]{\lineheight{1.25}\smash{\begin{tabular}[t]{c}$h$\end{tabular}}}}%
  \end{picture}%
\endgroup%

\end{center}
\caption{\label{xfig}Cycles around $v_c$}
\end{figure}

\begin{proof}[Proof of Lemma \ref{dlemma}]
We will show that if every point in the drawing of $C_i$ is in distance at least $x$ from $v$, then
every point in the drawing of $C_{i+1}$ is in distance $x+c$ from $v$. By induction, this is enough
to prove Lemma \ref{dlemma}.

Figure \ref{xfig} shows what happens for an Euclidean graph.
(We define Euclidean graph just like hyperbolic graph except the differing underlying geometry.) 
In this picture, the drawing of $C_i$ is a hexagon; 
cycle $C_i$ has six vertices around $v_c$. For each of these seven vertices the exclusion zone of radius $r/2$ around it is shown (in green).
It is clear from the picture (and the definition of Euclidean graph) that no point in the drawing of $C_{i+1}$ may fall into the gray/green region. 
All points in distance $\leq r/2$ from $C_i$
fall in the gray/green region, thus our claim is true for $c=h=r/2$.

In the hyperbolic plane the situation is slightly different. Because parallel lines work differently in the hyperbolic plane (they ``diverge''),
we have $c = h < r/2$, where $h$ depends on $r$ and $d$ (we have $h = \Theta(re^{-d})$ for large values of $d$). Still, our claim holds with $c>0$.
\end{proof}

\begin{corollary}\label{tdec}
Given a $(r,d)$-hyperbolic graph $G = (V,E)$, it is possible to find a tree decomposition of width $O(\log |V|)$ in polynomial time.
\end{corollary}
\begin{proof}
From Theorem \ref{nogrid} choose $t = O(\log |V|)$ such that $G$ does not contain a $k \times k$ grid. From Theorem \ref{robsey}, it is possible to find a tree decomposition of
width $5t-6 = O(\log |V|)$.
\end{proof}

\begin{corollary}
HECSP (and thus HLCSP) can be solved in polynomial time.
\end{corollary}

\begin{proof}
From Corollary \ref{tdec} we can find a tree decomposition $(V_T,E_T,B)$ of width $w = O(\log |V|)$. Then we use Dynamic Programming over $(V_T, E_T)$.
For every $b \in V_T$, we compute $s(w)$, the set of all possible colorings $c: X_b \in K$ such that there exists a coloring $c: X^+_b$ which extends $c$
and which satisfies all the constraints on edges in $X^+_b$. This can be computed straightforwardly in every case from Lemma \ref{cases}.
The whole algorithm works in $O(|V| \cdot |K|^w)$, which is polynomial in $|V|$.
\end{proof}

\begin{figure}[b]
\begin{center}
\subfig{0.23\linewidth}{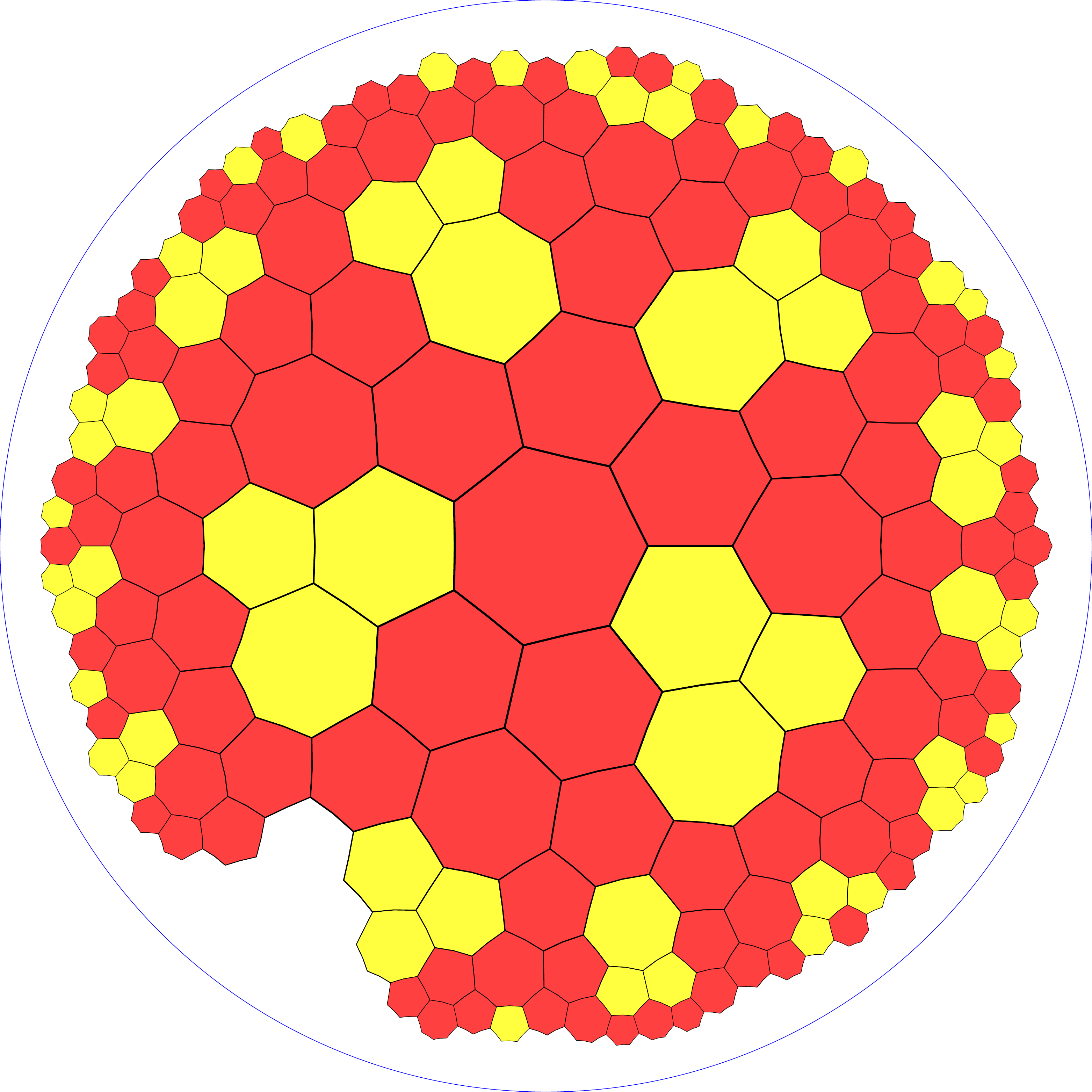}{two yellow}
\subfig{0.23\linewidth}{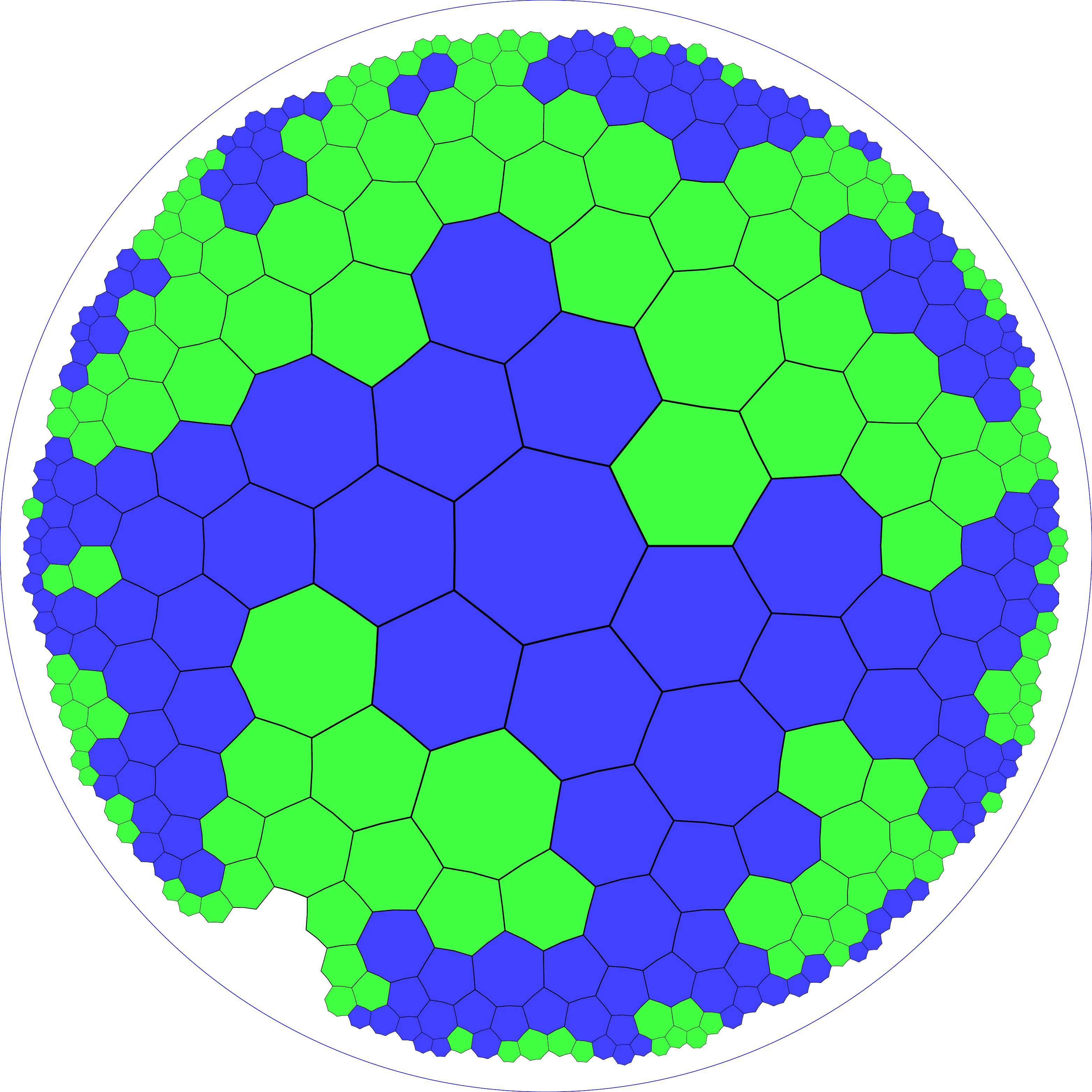}{one group}
\subfig{0.23\linewidth}{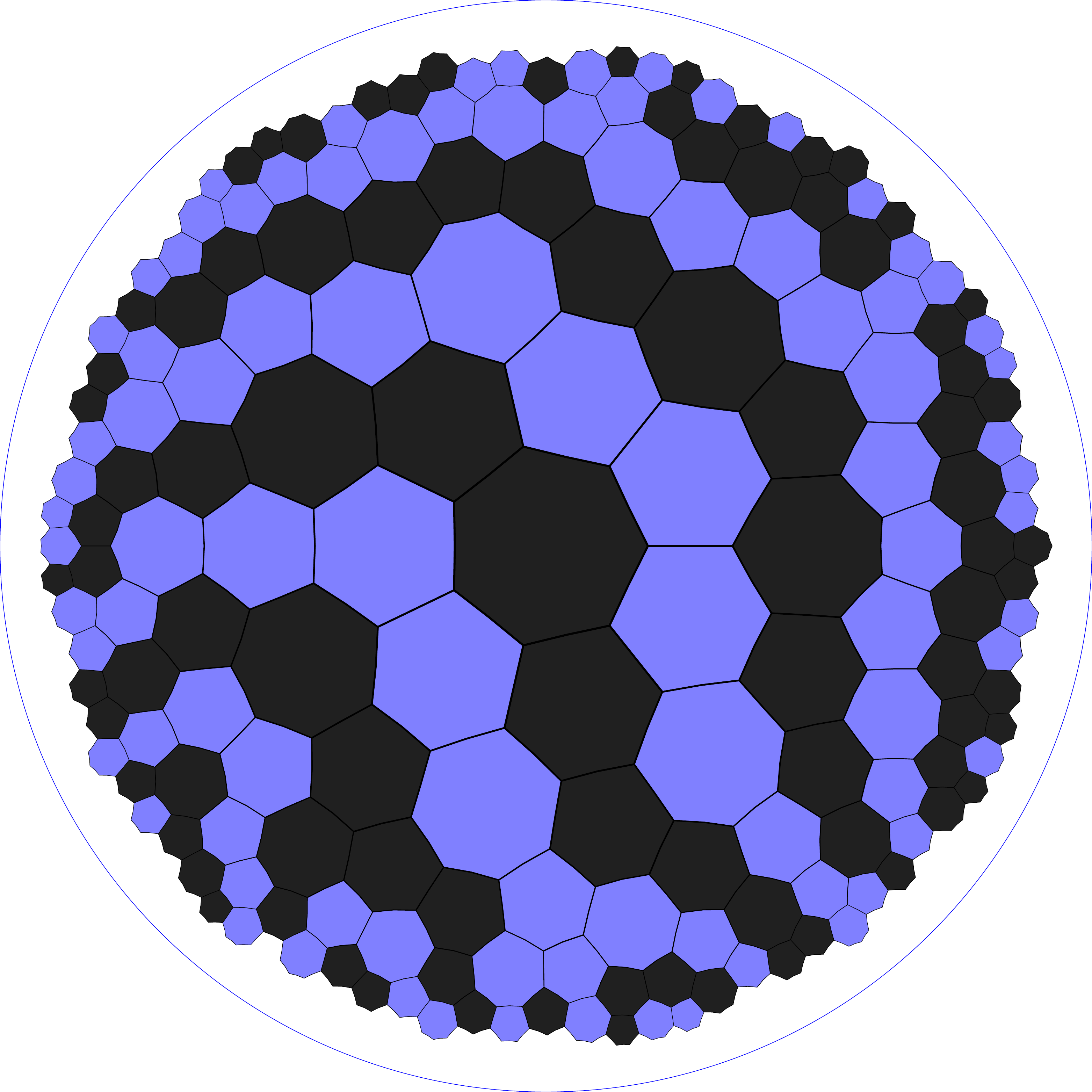}{two groups}
\subfig{0.23\linewidth}{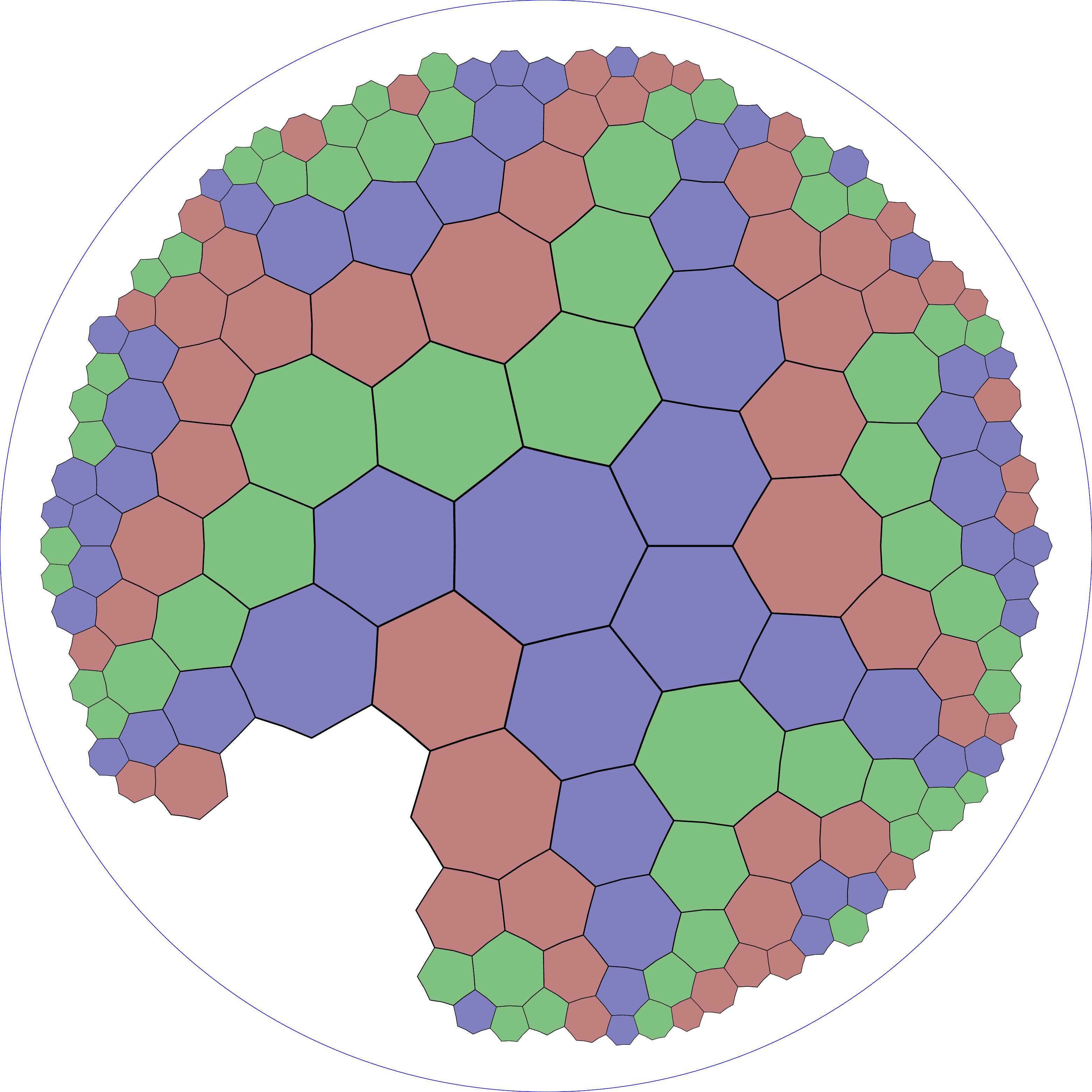}{four groups}
\end{center}
\caption{\label{consgen} Random colorings satisfying various constraints.
Since the degree of our polynomial is quite high, these pictures took about a minute to make.
Some cells have been removed from the full disk to reduce the treewidth. 
}
\end{figure}

\section{Conclusion}
We have shown that Minesweeper on hyperbolic tessellations can be solved in polynomial time. Our method is general: it works for any $(r,d)$-hyperbolic graph,
and for any problem based on satisfying local constraints. Other than solving puzzles, this may be applied to procedural content generation. For example, the Wave
Function Collapse (WFC) algorithm \cite{wfc,wfc_karth} is used in procedurally generated games to procedurally generate maps satisfying local constraints (e.g., a mountain
should not appear close to ocean, or roads should not branch nor end abruptly). In general finding out whether such constraints can be satisfied is NP-complete
(although this does not happen in typical PCG applications). Our algorithm can be adjusted to count the number of satisfying colorings, and to 
produce one of them, randomly chosen (Figure \ref{consgen}). 

Our results do not generalize to the three-dimensional hyperbolic space $\bbH^3$. This is because 
the Euclidean plane can be isometrically embedded in the three-dimensional hyperbolic space $\bbH^3$ as a horosphere.
Minesweeper played on a tiling which tessellates such a horosphere into squares is NP-complete.

HyperRogue also lets the player play Minesweeper in bounded hyperbolic manifolds, such as the Klein quartic or other Hurwitz manifolds. Hurwitz manifolds 
are quotient spaces of $\bbH^2$ which can be tiled with regular heptagons with angles $120^\circ$ (see Figure \ref{hrhex}b), and that are maximally symmetric,
i.e., there exists an isometry of the Hurwitz manifold $M$ into itself which map any heptagon with any orientation into any heptagon with any orientation.
Tessellations of quotient space are no longer planar, thus Theorem \ref{robsey} nor our proof of Theorem \ref{nogrid} is
no longer valid in quotient spaces of $\bbH^2$. Therefore, the complexity of Minesweeper on such manifolds does not follow from our results.
There are less symmetric hyperbolic manifolds into which the Euclidean square grid, or even higher-dimensional Euclidean grids, can be embedded \cite{hgeom},
and thus Minesweeper on them is NP-complete; however, highly symmetric manifolds have a more restricted structure.

\def\ext#1{#1}
\bibliography{minesweep}
\end{document}